\documentclass[DIV=calc, paper=letter, fontsize=11pt, twocolumn]{scrartcl}	 

\usepackage[protrusion=true,expansion=true]{microtype} 
\usepackage[svgnames]{xcolor} 
\usepackage[hang, small,labelfont=bf,up,textfont=it,up]{caption} 
\usepackage{fix-cm}	 

\usepackage{sectsty} 
\allsectionsfont{\usefont{OT1}{phv}{b}{n}} 

\usepackage{fancyhdr} 
\usepackage{lastpage} 

\usepackage{lettrine} 
\newcommand{\initial}[1]{ 
\lettrine[lines=3,lhang=0.3,nindent=0em]{
\color{DarkGoldenrod}
{\textsf{#1}}}{}}

\usepackage{graphicx} 
\usepackage{float} 
\usepackage[utf8]{inputenc}
\usepackage{amsmath,amsthm,amsfonts,algorithm,algorithmic}
\usepackage{amssymb}
\usepackage{xspace}
\usepackage{comment}
\usepackage{color}
\usepackage{cleveref}
\usepackage{fancybox}

\usepackage{duerer}

\floatstyle{boxed}
\newfloat{boxing}{tbph}{boxes} 
\floatname{boxing}{Box} 

\newcommand{\kmer}{$k$-mer\xspace}
\newcommand{\kmers}{$k$-mers\xspace}

\newcommand{\setof}[2]{\{#1\;:\;#2\}}

\newcommand{\kmomer}{$(k-1)$-mer\xspace}
\newcommand{\kmomers}{$(k-1)$-mers\xspace}
\newcommand{\kspec}{\text{$\mathrm{sp}^k$\xspace}\xspace}
\newcommand{\twospec}{\text{$\mathrm{sp}^2$\xspace}\xspace}

\newcommand{\kmospec}{\text{$\mathrm{sp}^{k-1}$\xspace}}

\newcommand{\dna}[1]{\text{\normalfont\duttfamily #1}}

\theoremstyle{plain}

\newtheorem{theorem}{Theorem}
\newtheorem{problem}{Problem}

\newtheoremstyle{westside}
  {1pt}
  {10pt}
  {\itshape}
  {0pt}
  {\bfseries}
  {}
  { }
  {\thmname{#1}\thmnumber{ #2}\thmnote{ (#3)}}

\theoremstyle{westside}

\newtheorem*{intproblem*}{}

\def\EnableMNotes{0}

\ifnum\EnableMNotes=1
\newcommand{\mynote}[1]{\footnote{\textsf #1}}
\else
\newcommand{\mynote}[1]{}
\fi

\usepackage{titling} 

\newcommand{\HorRule}{\color{DarkGoldenrod} \rule{\linewidth}{1pt}} 

\pretitle{\vspace{-30pt} \begin{flushleft} \HorRule \fontsize{23}{23} \usefont{OT1}{phv}{b}{n} \color{DarkRed} \selectfont} 

\title{Modeling Biological Problems in Computer Science: A Case Study in Genome Assembly}

\posttitle{\par\end{flushleft}\vskip 0.5em} 

\preauthor{\begin{flushleft}\large \lineskip 0.5em \usefont{OT1}{phv}{b}{sl} \color{DarkRed}} 

\author{Paul Medvedev } 

\postauthor{\footnotesize \usefont{OT1}{phv}{m}{sl} \color{Black} 
\\
Pennsylvania State University \\ 
pashadag@cse.psu.edu

\par\end{flushleft}\HorRule} 

\date{\today} 

\begin{document}
\maketitle
\thispagestyle{fancy} 

\initial{A}\textbf{s computer scientists working in bioinformatics/computational biology, 
	we often face the challenge of coming up with an algorithm to answer a biological question.
	This occurs in many areas, such as variant calling, alignment, and assembly.
	In this tutorial, we use the example of the genome assembly problem to demonstrate how to go from a question in 
	the biological realm to a solution in the computer science realm. 
	We show the modeling process step-by-step, including all the intermediate failed attempts.\footnote{Please 
	note this is not an introduction to how genome assembly algorithms work and, if treated as such, would be incomplete and unnecessarily long-winded.}
}

\section{The biological question}
First, we must understand the biological question. 
We must rely on a biologist to explain it to us, in a simplified way which we can understand.
This can happen directly, through communication with a biologist, or it can happen indirectly, 
by reading articles in biological journals.
Let us begin with a hypothetical description that might have been provided to us:

\begin{problem}[Biological problem]\label{pr:bio}
Inside the nucleus of every cell is one or more DNA molecules, altogether called the genome.
A DNA molecule is a chain composed of smaller elements called nucleotides (abbreviated as nt).
There are four nucleotides: cytosine (\dna{C}), guanine (\dna{G}), adenine (\dna{A}), and thymine (\dna{T}).
Biologists are interested in knowing the sequence of this chain, but, unfortunately, 
there is of yet no technology that could report this full sequence.
What we do have is called {\em sequencing technology}, and it works as follows.
The DNA is extracted from each cell and then amplified (i.e. many copies are made).
The DNA is then fragmented into smaller pieces, and the initial sequence of each smaller piece is determined (with occasional errors).
The machine outputs a collection of these short sequences, called {\em reads}.
The number of reads and their length varies by the specific technology and experiment.
One typical scenario is to have 1 million reads of length 200nt for a genome of length $5$ million nucleotides.
Given such a collection of reads, the biologist would like to know the full sequence of the genome.
This is loosely referred to as the problem of {\em genome assembly}.
\end{problem}

A biologist would immediately point out that calling the above description a ``simplification'' is a bit of a euphemism.
In fact, there are two types of simplifications in the above description.
The first are factual inaccuracies. 
For example, a DNA molecule is not a chain composed of nucleotides, 
it is in fact two chains of nucleotides, with each chain being the Watson-Crick complement of the other\footnote{Recall 
from high-school biology that DNA is double-stranded, with each nucleotide form one strand being paired with one on the other.}.
Another example is that the amplification process is not perfect and can occasionally create DNA molecules that do not originally occur in the cell.
The second type of simplification is the omission of important information.
For example, we have left out the fact that the DNA molecule actually has a 3D structure, and this structure can influence how the fragmentation happens. 
Another omission is that the machine can output quality values for each nucleotide, representing the probability of an error.
In addition to intentional simplifications, there are also unknown factors that cannot be captured in a problem description.
For example, the technician performing the experiment may have made a mistake in the experimental protocol, resulting in unpredictable effects.

Simplifications are unavoidable.
As computer scientists, it would be time-prohibitive to understand all the details of the sequencing process.
In describing a problem to a computer scientist, a biologist must implicitly decide what details are most relevant and omit the ones that she thinks are not.
Even if a biologist wants to describe every single detail, 
there are details at the chemical level that even she does not fully grasp.
And even if she is also an expert in chemistry, there would be details at the level of physics that she would not grasp. 

Fundamentally, 
when a biological data generation process is described in words (even in excruciating detail),
the words are, by definition, a simplification of reality.
Even if we propose a perfect solution to the biological problem, as we understood it, 
this could still be a failure when applied to real data.
Hence, we must never be too confident in the correctness of any solution we propose for a stated biological problem. 
Keeping this caveat in mind, let us try to find a solution.

\section{A well-formulated problem}\label{sec:well}
It might be tempting to immediately start designing an algorithm to solve~\Cref{pr:bio}.
But, computer science is best at providing solutions to problems that are stated with a certain mathematical rigor.
The description of~\Cref{pr:bio} is not at a mathematical level which would allow us to harness the power of computer science.
Therefore, we must first formulate a problem which is stated in the language of computer science --- we call such a problem {\em well-formulated}.
Box~\ref{box:well} lists some principles of a well-formulated problem. 
Let us make our first attempt at a formulation.
\begin{intproblem*}
	\hspace{1pt} \\
	\textbf{ Input:} A collection of strings $S=\{s_1, \ldots, s_m\}$ generated by a sequencing experiment from a genome.\\
		\textbf{ Output:} The genome that generated $S$. 
\end{intproblem*}
Here, we have made explicit the input and output requirements of our problem.
We have also replaced the imprecise notion of a ``short sequence,'' used to describe a read in~\Cref{pr:bio}, with the 
precise formality of a string.
Nevertheless, this problem falls short of being well-formulated for numerous reasons.
A useful sanity check to detect a bad formulation is to imagine giving it to another computer scientist, 
without giving any of the biological background. 
Would she be able to solve it?  
No --- she has no idea what the terms ``sequencing experiment,'' ``generated,'' or ``genome'' mean.
In other words, the problem is not self-contained.
Let us try again.
\begin{intproblem*}
	\hspace{1pt} \\
	\textbf{ Input:} A collection of strings $S=\{s_1, \ldots, s_m\}$, such that each string is a substring of an unknown string called the genome. \\
			\textbf{ Output:} The unknown genome.
\end{intproblem*}

We have made some improvements here --- we are no longer referring to those mathematically undefined terms. 
But we fail another sanity check.
A formulation must implicitly define a way to check if the output of a program that solves it is correct. 
For example, if $S=\{\dna{CGG}, \dna{AAC}\}$, is an output of \dna{CGGAAC} correct? That could have been the unknown string called the genome, but 
perhaps the genome was \dna{AACGG}. 
Our problem is ambiguous in this regard and does not give us a way to verify whether an algorithm is correct.
Undeterred, we can fix the formulation as follows: 

\begin{problem}[well-formulated problem]\label{pr:cs2}
	\hspace{1pt} \\
	\textbf{ Input:} A collection of strings $S=\{s_1, \ldots, s_m\}$. \\
			\textbf{ Output:} A string $g$ that is a {\em common superstring} of $S$, i.e. for all $i$, $s_i$ is a substring of $g$.
\end{problem}

The difference in this formulation is subtle but important. 
It makes no reference to an ``unknown genome.'' 
Instead, it states clearly that an output is correct if and only if it is a common superstring.
Thus, both \dna{AACGG} and \dna{CGGAAC} are correct solutions.
In fact, we finally have a well-formulated problem!

\begin{Sbox}
\begin{minipage}{3.2in}
	\begin{enumerate}
		\item The input/output requirements are defined.
		\item All terms are mathematically defined.
		\item The description is self-contained.
		\item It is possible to verify if a solution is correct.
	\end{enumerate}
\end{minipage}
\end{Sbox}

\begin{boxing}[t]
\begin{center}
	\TheSbox
\end{center}
\caption{Some principles of a well-formulated problem.}
\label{box:well}
\end{boxing}

\section{Simplifying assumptions}\label{sec:simplify}

In order to arrive at~\Cref{pr:cs2}, 
we made some simplifying assumptions about the biological process described in~\Cref{pr:bio}.
For example, we ignored the fact that the reads may have errors. 
If a read has an error in it (e.g. one position is a \dna{C} instead of an \dna{A}), then 
it is no longer a substring of the genome.
We already argued how the simplifications in~\Cref{pr:bio} were inevitable.
But the simplification of~\Cref{pr:cs2} was by our choice,
and it is not clear that a solution to~\Cref{pr:cs2} will work on real data.
So why did we create this simplification then?

We hope that a simple formulation, even if it is inaccurate, will lead to an algorithm that is useful in practice.
Generally speaking, a formulation is useful if it leads to an algorithm that is efficient and accurate 
when applied to real data.
Making the model more accurate would also make it more complicated, representing  a trade-off. 
On one hand, a formulation that is too simple risks leading to an algorithm that is not useful in practice.
On the other hand, a formulation that is too realistic risks making it impossible to solve the problem.
How do we know where to draw the line then? There is no simple rule -- that is part of the research challenge.
When faced with situations like this, computer scientists have three choices:
\begin{enumerate}
	\item Use a simplified formulation and hope that its inaccuracy does not affect the accuracy of the algorithm on real data. 
		In other words, we hope that the solution to our problem formulation will work well 
		when applied to real data, regardless of the simplifications.
	
	\item  Make the formulation more complex by removing the unrealistic assumptions. 
		In our case, we would include the possibility of errors by 
		loosening the superstring requirement 
		in Problem~\ref{pr:cs2} to allow mismatches.
		In going with this option, 
		we hope that we can still come up with an efficient algorithm, in spite of the model's complexity.
	\item 
		Break the problem into multiple sub-problems (i.e. modularize the problem formulation).
		In our case, we could formulate two problems instead of one.
		The first would be to take erroneous reads and correct them, and 
		the second would be to take corrected reads and find the common superstring (i.e.~\Cref{pr:cs2}).
		The second formulation makes the simplifying assumption,
		while the first is designed to ensure the assumption holds.
		Modularization has the potential to avoid complex formulations while avoiding simplifying assumptions.
		However, it is often not possible to transform the data so that the simplifying assumption holds.
\end{enumerate}
Knowing which option to choose is an art rather than a science.
The correct choice is often arrived at with trial-and-error and involves extensive experiments on both simulated and real data.
We will discuss this more thoroughly in~\Cref{sec:exp}.
Meanwhile, in~\Cref{pr:cs2}, we have chosen the first option, though we leave open the possibility to later switch to option three if necessary.

When solving our problem, we must keep in mind that our algorithm will likely live inside of a bioinformatics pipeline.
Both upstream and downstream algorithms will introduce errors into the data stream,
and all algorithms within the pipeline must be robust to the presence of such errors.
The benefits of an algorithm that perfectly solves a formulation are often nullified once it becomes part of such a pipeline,
because 1) a problem formulation can never capture all the sources of errors
and 2) errors will anyway be introduced by downstream algorithms.
The ultimate test of an algorithm is how it performs on real data, as opposed to its provable properties with respect 
to the problem formulation.

\section{Optimization criterion}
Problem~\ref{pr:cs2} is a well-formulated problem and so we can attempt a solution.
Can we solve it? Yes! It turns out to be very easy -- too easy in fact. 
One very efficient algorithm is to just output the concatenation of all the strings: $s_1s_2\cdots s_m$. 
Unfortunately, this solution is not very useful, because it is unlikely to have been the genome that generated the reads.
One way to see this is to observe that there are many other solutions. 
For example, we could concatenate the strings in any order we want.
Any such concatenation would be an equally valid solution to our formulation. 
In fact, there are infinitely many solutions, because if we concatenate any solution with itself, it is still a valid solution.
That makes it unlikely that if we just pick an arbitrary one, it will be the correct biological one.
So, while the problem is well-formulated, it is not very useful.

The issue with~\Cref{pr:cs2} is that while there are many solutions, 
there is no attempt to distinguish how one is better than another.
What we need instead is an {\em optimization problem}.
In an optimization problem, there is a set of {\em feasible solutions} and 
a function that gives a cost or a score to each feasible solution.
An {\em optimal solution} is a feasible solution with the lowest cost or highest score.

If we think of the common superstrings as the set of feasible solutions,
what kind of cost function could we use to differentiate between them?
Let us consider the example of $S=\{\dna{ACG}, \dna{CGT}\}$.
Both \dna{ACGT}, \dna{ACGCGT}, and \dna{ACGTTTTTCGT} are common superstrings.
Intuitively, though, one would probably pick \dna{ACGT} as the best solution because it is the shortest.
This intuition reflects the {\em parsimony principle}, which states that the simplest explanation is the correct one.
We can reformulate our problem as follows:
\begin{problem}[shortest common superstring]\label{pr:scs}
	\hspace{1pt} \\
	\textbf{ Input:} A collection of strings $S=\{s_1, \ldots, s_m\}$.\\
		\textbf{ Output:} A string $g$ that is a shortest common superstring of $S$.
\end{problem}
Notice that this is an optimization problem where the cost function is the length of the feasible solution.
Looking at our example, 
the optimal solution is in fact \dna{ACGT} and it is unique.
This is great, but, in the general case, 
do we have the desired property that there is a limited number of optimal solutions?
The answer is not immediately obvious, and we will revisit it in~\Cref{sec:contigassembly}.

\section{Computational complexity}\label{sec:complexity}
One issue we have not yet addressed is the computational complexity of solving the problem we formulate\footnote{For 
readers unfamiliar with computational complexity, it is considered useful to separate problems 
that have a polynomial-time algorithm from those for which such an algorithm is not known. The first class of problems are said to be {\em tractable}.
NP-hard problems are a type of problem that belong to the second class. 
There can be non-polynomial-time algorithms for NP-hard problems, but such algorithms typically do not scale well for large problems.}.
While coming up with a formulation, a good practitioner will always keep in mind whether 
the formulated problem can be efficiently solved using well-known techniques. 
What is its relationship to other well-studied theoretical problems?
In this case, \Cref{pr:scs} is a well-studied computer science problem called the shortest common superstring (SCS) problem. 
Its computational complexity turns out to be known; it is NP-hard, meaning it is unlikely that an exact polynomial-time algorithm exists.
This does not necessarily make a problem formulation bad. 
There are whole books describing available techniques to solve NP-hard problems. 
For example, the size of the input may be small enough to allow an asymptotically slow but exact algorithm.
Or, perhaps the worst-case inputs do not occur in practice or 
a heuristic algorithm\footnote{{\em Heuristics} are algorithms that do not always give the correct solution but are typically efficient.}
will work well on real instances.

Of course, ideally we could have a formulation that is tractable (i.e. can be solved in polynomial time).
One way that this might still be achieved is to add or remove constraints from the definition.
As we will see in~\Cref{sec:toolbox}, simplifying our assumptions, and thereby removing constraints,
can make a problem tractable.
Adding new constraints can also have the same effect\footnote{The 
famous satisfiability problem (SAT) is a good example, for those readers who are familiar with it.
SAT is NP-hard, but, if each clause has at most two variables, it becomes solvable in polynomial time.}.
In our case, Problem~\ref{pr:scs} ignores the fact, 
which is clearly stated in Problem~\ref{pr:bio}, 
that the strings are over an alphabet of size 4.
Could Problem~\ref{pr:scs} be tractable if this restriction is added? 
Unfortunately for us, a quick internet search reveals that 
the SCS problem is NP-hard even when the alphabet is binary.

\section{Simulations}
Putting aside the issue of computational complexity, let us go back again and ask ourselves:
if someone gave us a black-box algorithm for~\Cref{pr:scs}, 
would the output in fact reconstruct the genome (i.e. solve Problem~\ref{pr:bio})?
To find the answer, we need to test a solution to the SCS problem on a dataset where we know the answer.
We take a known genome sequence, e.g. that of the influenza virus, and do an idealized simulation of the sequencing process.
That is, we take substrings of length 200 from every position of the genome. 
Next, we need to implement an algorithm for SCS. 
First, we might try an exponential-time brute force algorithm, but we would find that it would not 
complete a run on our dataset.
Instead, a quick literature search reveals the existence of a greedy heuristic.
Initially, we start with a string made up from one arbitrary read.
To expand the string, we pick a read that overlaps our current string by the maximum possible among all reads.
We then extend our string using this read, and repeat until all reads are accounted for.
After implementing this heuristic, we run it on our dataset.

When comparing our solution to the influenza genome, we find that they are different.
Moreover, our solution is shorter than the influenza genome. 
After validating that the solution is indeed a common superstring of the reads,
we start to investigate the reason for the missing sequences. 
What we find is that some sequences which appear multiple times in the virus genome appear only once in our solution.
We follow up with a domain expert -- a biologist in this case -- to check if our solution makes sense.
After a careful analysis of our solution, she observes that there is something strange.
Genomes are typically full of repetitive elements -- DNA sequences that appear in multiple locations of the genome.
However, our solution has a very low repeat content.
Could it be that the problem is due to using a heuristic, instead of an exact algorithm?
After thinking more about the problem, we have a great insight that we capture in the following theorem:

\begin{theorem}\label{thm:scs}
	Let $S=\{s_1, \ldots, s_m\}$ be a set of strings of length $\ell$.
	Let $g$ be a shortest common superstring of $S$.
	Then, for any string $r$ with length greater than $2\ell - 2$, $r$ can be a substring of $g$ in at most one location.
	In other words, a SCS of $S$ does not contain any repeat of length greater than $2\ell - 2$.
\end{theorem}
\begin{proof}
	Assume for the sake of contradiction that $g$ contains a repeat $r$, with $|r| > 2\ell - 2$.
	Since the length of $r$ is greater than $2\ell - 2$, we can uniquely write it as a concatenation of three strings $r=xyz$, 
	where $|x| = |z| = \ell - 1$ and $|y| \geq 1$.
	Construct another string $g'$ that is exactly the same as $g$ except that the second occurrence of $r$ in $g$ is replaced by $xz$.
	In other words, we delete $y$ from the second occurrence of $r$.

	Now, we claim that $g'$ is a common superstring of $S$. 
	For any string $s \in S$, consider the range in $g$ where it is a substring.
	If that range does not overlap the location of the deleted $y$, 
	then $s$ will also be a substring of $g'$.
	If the range overlaps, then $s$ must be a substring of $xyz$ (and hence of $g'$).
	This is because the length of $s$ is only one character more than the length of $x$ or of $z$.

	Thus, $g'$ is a common superstring of $S$.
	Also, $g'$ is strictly shorter than $g$ because we deleted a non-empty string ($y$) from $g$ to obtain $g'$.
	This is a contradiction of the fact that $g$ was the shortest common superstring.
\end{proof}

\Cref{thm:scs} indicates that the problem does not lie with our use of a heuristic.
Instead, it is a serious shortcoming of the SCS formulation itself. 
Specifically, it states that a SCS would not faithfully represent repeats, 
which is known as the problem of over-collapsing of repeats.
As a result of~\Cref{thm:scs}, a SCS will not be correct for most genomes. 
It turns out that our intuition for using parsimony was not completely correct.
Problem~\ref{pr:scs} is not a useful problem formulation.

\section{Adding constraints}
We need a problem formulation that avoids collapsing repeats.
Part of the issue is that we allow ourselves a lot of leeway when reconstructing the genome.
We are allowed to create parts of the genome that are not supported by any reads.
For example, in the proof of~\Cref{thm:scs}, the common superstring $g'$ contains the junction between $x$ and $z$, 
which might not be supported by any read. 
Consider the following example, where the reads are all the substrings of length $\ell=3$ generated from this genome:
\begin{align*} 
	g_\text{true} = \dna{A\underline{ATTCCAG}CTG\underline{ATTCCAG}T} \\
\end{align*}
Notice that there is a repeat of length $>2\ell - 2$, which we underline.
The SCS solution is
\begin{align*}
	g_\text{scs} = \dna{AATTCCAGCTGA\underline{TA}GT}
\end{align*}
The underlined characters in $g_\text{scs}$ indicate the location of the unsupported junction.
That is, the string \dna{TA} does not occur in the reads.
We could consider adding a constraint that any 2-mer in our solution must be contained in some read
(a {\em \kmer} is just a string of length $k$).
The hope would be that such a constraint would prevent the proof of~\Cref{thm:scs} from going through and 
could force the solution to more faithfully reconstruct the repeats.
Formally:
\begin{intproblem*}
	\hspace{1pt} \\
	\textbf{ Input:} A collection of strings $S=\{s_1, \ldots, s_m\}$.\\
		\textbf{ Output:} A shortest string $g$ that is a common superstring of $S$ and satisfies that,
		for all $1 \leq i < |g|$, there exists an $1 \leq j \leq m$ such that
		$g[i..i+1]$ is a substring of $s_j$.
\end{intproblem*}

We use the notation $g[a..b]$ to refer to the substring of $g$ 
whose first position is $a$ and whose last position is $b$.
Now, let us check if our new constraint is too strong by thinking about 
the typical sequencing experiment described in Problem~\ref{pr:bio}. 
The probability that there exist two consecutive positions of the genome which are not spanned by a read can be
computed from the length of the genome and the number and length of reads, if we assume that the reads are sampled uniformly\footnote{In 
other words, each location of the genome is equally likely to generate a read.} at random from the genome.
This probability is extremely small, 
hence our assumption that every 2-mer of the genome is included in some read is reasonable
(subject to the uniformity assumption, which we will revisit in~\Cref{sec:exp}).

Let us go back to our example---we should make sure that the new constraint works to give us the correct answer.
Indeed, the solution $g_\text{scs}$ would no longer be valid, since the underlined 2-mer is not part of any read.
The new solution is:
\begin{align*}
	g_\text{sol} = \dna{AATTCCAGCTGAT\underline{G}AGT} 
\end{align*}
Here, the new character relative to $g_\text{scs}$ is underlined.
Notice that even though every 2-mer in the solution exists in the reads,
$g_\text{sol}$ is still shorter than the true genome.
However, if we consider 3-mers instead of 2-mers, then 
\dna{ATG} and \dna{GAG} are not part of any read.
We can verify that if we require every solution 3-mer to be present in the reads, 
then $g_\text{true}$ finally becomes the unique solution to our problem.
More generally, then, could increasing the value of $k$ for which we require each solution \kmer to be included in the reads
help us reconstruct the true genome?

Perhaps, but we must take care that the chance of a genomic \kmer not being sampled is small.
Let us go back to the typical scenario described in~\Cref{pr:bio}.
For simplicity, consider the case where our one million reads are distributed evenly throughout the five million nucleotide genome, 
i.e. there is precisely one read sampled every fifth position.
Since each read is 200nt long, then every range of width 195nt is spanned by a read.
Even in the more general case when reads are sampled uniformly at random, 
we still expect ranges much wider than 2nt to be spanned.
In fact, for a given $k$, read length, genome length, and number of reads, 
we can derive the probability that a range of size $k$ in the true genome is not spanned by a read.
We will omit the derivation here, but assume that $k$ can be chosen so that this probability is neglible.

\begin{figure*}[h]
	\center
	\includegraphics[scale=0.5]{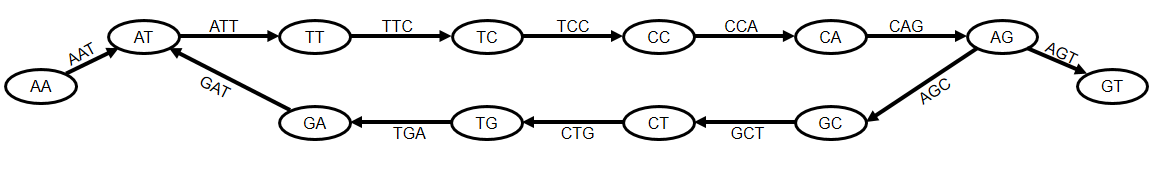}
\caption{An example of a de Bruijn graph of order 3. It is constructed from all substrings of length three from the string 
	$g_\text{true}=\dna{AATTCCAGCTGATTCCAGT}$.
\label{fig:dbgscs}}
\end{figure*}

We are now ready to formalize our stronger constraint. 
Given a string $s$, 
the \kmer-spectrum of $s$, denoted by \kspec, is the set of all $k$-long substrings of $s$, i.e. 
$\kspec(s) = \setof{x}{x \text{ is a substring of $s$ and $|x| = k$}}$.
For a collection of strings $S$,
the \kmer-spectrum is naturally defined as the set of all \kmers of $S$, i.e.
$\kspec(S) = \bigcup_{s\in S} \kspec(s)$.
Our last problem formulation could have been stated as requiring that the output satisfies $\twospec(g) \subseteq \twospec(S)$.
Now we generalize this to \kmers:
\begin{problem}[spectrum-based] \label{pr:dbg1}
	\hspace{1pt} \\
	\textbf{ Input:} A collection of strings $S=\{s_1, \ldots, s_m\}$, and an integer parameter $k$.\\
		\textbf{ Output:} A shortest string $g$ that is a common superstring of $S$ and satisfies  $\kspec(g) \subseteq \kspec(S)$.
\end{problem}
Looks good, but how in the world do we solve this?

\section{Changing the toolbox}\label{sec:toolbox}
Solving computer science problems from scratch can be a daunting task. 
While theoretical computer scientists will appreciate the beauty and challenge of a brand new problem,
our best bet is to try to connect it to something that has been well-studied.
So far, our problem formulations belong to the realm of {\em stringology}, but
it is not clear how to solve~\Cref{pr:dbg1} using the stringology toolbox.
In such cases, it may be useful to change perspectives and see if a different toolbox might apply.
We will try graph theory, 
which offers an intuitive visual representation and is a deeply studied area.

The idea of strings with a constrained spectrum is closely related to something called the {\em de Bruijn graph}.
The $k^\text{th}$ order de Bruin (directed) graph of a set of strings $S$ is denoted by $DB^k(S)$.
Its vertices correspond to $\kmospec(S)$ and the edges correspond to $\kspec(S)$. 
Specifically, for a \kmer $x$, there is an edge from $x[1..k-1]$ to $x[2..k]$.
Note that even if a \kmomer or a \kmer appears more than once in $S$, 
it is still represented by only one vertex or edge, respectively.
An example de Bruijn graph is shown in~\Cref{fig:dbgscs}.

A {\em walk} in a directed graph is a sequence of edges such that every edge leaves the vertex which the previous edge enters;
a walk is {\em edge-covering} if it visits every edge at least once.
Notice that any walk in $DB^k(S)$ spells a string whose \kmer spectrum is exactly the set of edges that the walk visits and,
therefore, is a subset of $\kspec(S)$.
Likewise, any string whose \kmer spectrum is a subset of $\kspec(S)$ is spelled by a unique walk in the graph.
This walk visits all the \kmomers of the string, in their respective order.
Furthermore, observe that if a walk spells a common superstring of $S$, then it is edge-covering.
We can therefore restate Problem~\ref{pr:dbg1} as follows:

\addtocounter{problem}{-1}
\begin{problem}[spectrum-based, restated]
	\hspace{1pt} \\
	\textbf{ Input:} A collection of strings $S=\{s_1, \ldots, s_m\}$, and an integer parameter $k$.\\
		\textbf{ Output:} 
		A shortest string $g$ that is a common superstring of $S$ and is 
		spelled by an edge-covering walk in $DB^k(S)$.
\end{problem}

The de Bruijn graph gives us another way to understand why Problem~\ref{pr:dbg1} helps alleviate the 
problem of over-collapsing of repeats.
Our solution is constrained to walk and cover the graph, and cannot simply include arbitrary characters or jump between 
different vertices.
The reader can use~\Cref{fig:dbgscs} to verify that $g_\text{scs}$ and $g_\text{sol}$ are not spelled by a walk in the
de Bruijn graph, for our example.

In addition to over-collapsing of repeats, Problem~\ref{pr:scs} suffered from being intractable.
What is the computational complexity of Problem~\ref{pr:dbg1}?
The interested researcher would attempt to find if the answer is known, and, in its absense,
would attempt to prove NP-hardness or find an algorithm herself.
In this case, she would discover that the problem is NP-hard by a reduction from Problem~\ref{pr:scs}.
Though we do not include the proof here, it suggests that the difficulty lies in fulfilling the superstring requirement.

As we discussed in~\Cref{sec:complexity}, one way to deal with NP-hardness is to try adding or removing constraints
to make the problem tractable.
Since our interested researcher is quite talented, she would have come across 
the Chinese Postman Problem during her search for an algorithm.
The Chinese Postman Problem is to find the shortest cyclical edge-covering walk, and the problem is known to 
be solvable in polynomial time.
It seems similar to Problem~\ref{pr:dbg1}, but it does not include the superstring requirement.
Let us then try dropping this requirement, though retaining the weaker  requirement that the solution includes all the \kmers from the reads.
\begin{problem}[de Bruijn graph-based] \label{pr:dbg2}
	\hspace{1pt} \\
	\textbf{ Input:} A collection of strings $S=\{s_1, \ldots, s_m\}$, and an integer parameter $k$.\\
		\textbf{ Output:} 
		The string spelled by a shortest edge-covering walk of $DB^k(S)$.
\end{problem}

Indeed, this is a minor variation of the Chinese Postman Problem, and 
we can solve it in polynomial time.
Changing our toolbox has allowed us to make the connection to a well-studied problem and 
identify a formulation which can be solved in polynomial time.

However, we had to make an additional simplification, by dropping the superstring requirement.
As a result, we might now have a read that is not represented anywhere in the solution. 
Our intuition has guided us towards choosing the first option from~\Cref{sec:simplify}, that is, 
we think this simplification is worth making.
But how do we verify that our intuition has not mislead us again and that our formulation is useful?

\section{Judging usefulness: experimental evaluation}\label{sec:exp}

In the process of formulating a problem, we make several simplifying assumptions. 
None of these hold perfectly in the real world.
Thus, even if we have the best algorithm, it can still fail miserably when faced with real or 
simulated data which violates these assumptions.
The ultimate test of an algorithm's usefulness is, therefore, how it fairs when the assumptions are not satisfied.
This is referred to as the {\em robustness} of the algorithm.

In our case, the formulation of Problem~\ref{pr:dbg2} is based on several assumptions. 
One assumption is that there are no errors in the reads.
We could of course modularize the problem (i.e. the third choice in~\Cref{sec:simplify})
and run an error-correction algorithm as a pre-processing step.
Even so, it will not be perfect and some errors will remain.
Another assumption is that every \kmer from the genome is included in some read.
Recall that we chose a $k$ value such that the probability of missing a \kmer is negligible,
under the assumption that the reads are sampled uniformly at random from the genome.
However, this is sometimes not the case, as certain genomic regions are more difficult to sequence 
due to their 3D structure.
Thus, some regions will have very little or no sampled reads, resulting in what are called {\em coverage gaps}.

Testing an algorithm on real data will ultimately uncover if it is not useful.
However, in many cases, we do not know the ground truth in real data. 
This makes measuring the accuracy of our algorithm difficult.
Even when the ground truth is available, a test on real data makes it hard to identify the algorithm's point of failure.
The first stumbling block may be that the problem formulation is flawed even if the data satisfies all our assumptions.
The second is that the algorithm is not robust with respect to violation of the assumptions.
The third is that even if it is robust to the assumptions known to us, 
there are aspects of real data which we do not understand and which are problematic for the algorithm.

Evaluation of an algorithm should therefore proceed in three stages.
\begin{enumerate}
	\item We begin by generating simulated data under the assumptions we have made.
		For example, our influenza simulation was of this kind, since we did not include any errors or coverage gaps.
		This gives us a controlled environment in which we can verify that the algorithm passes the first stumbling block.
	\item Next, we use simulated data which models the real data, to the extent we understand it.
		For example, we may generate reads with errors or with coverage gaps, mimicking what we observe in real data.
		This gives us a controlled environment in which to verify that the algorithm passes the second stumbling block.
	\item Finally, we run the algorithm on real data.  If it does not perform well, we must investigate more carefully to understand which aspect of the data we are not modeling correctly.
\end{enumerate}

\section{Moving the goalposts}\label{sec:contigassembly}
Buoyed by the ability of Problem~\ref{pr:dbg2} to solve the small example of~\Cref{fig:dbgscs}, 
we implement our polynomial time algorithm so that we can check its usefulness.
To do so, we follow the advice of the previous section and first run the algorithm on our simulated influenza data.
To our profound dismay, we again find that the solution does not match the true genome.
When we examine the solution closer, we find that it contains many long correct segments, but 
they are arranged in the wrong order.
The problem, it turns out, is that there are usually many optimal solutions 
(recall that we mistakenly thought we solved this issue in Problem~\ref{pr:scs}).
By looking at de Buijn graphs from simulations, we observe an interesting pattern.
Consider the example in~\Cref{fig:bubbles}.
The graph contains three ``bubbles,'' and 
there are $2^3$ optimal solutions, depending on the order we choose to traverse the top/bottom part of each bubble.
For example, both
$AB_1CD_1EF_1GB_2CD_2EF_2H$ and
$AB_2CD_1EF_2GB_1CD_2EF_1H$ 
are optimal solutions.
This example can be generalized to $n$ bubbles, creating exponentially many ($2^n$) optimal solutions.
Clearly, picking an arbitrary one is unlikely to give us the correct genome.

\begin{figure*}
	\center
	\includegraphics[scale=0.5]{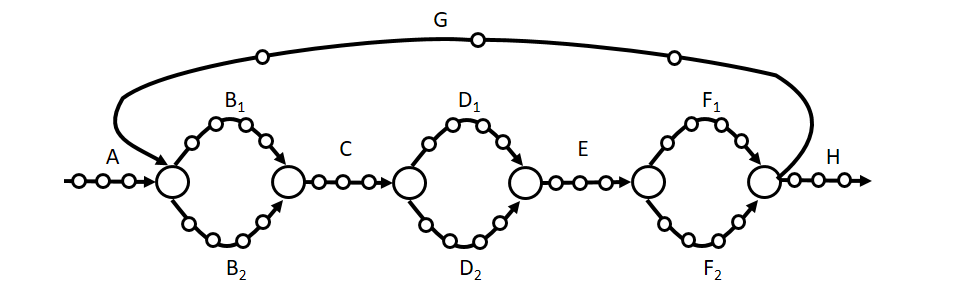}
	\caption{An example of a de Bruijn graph for which Problem~\ref{pr:dbg2} has exponentially many optimal solutions.
\label{fig:bubbles}}
\end{figure*}

Could the culprit be our decision to drop the superstring requirement?
No, because if the length of the paths connecting consecutive bubbles (C, E, and G) is larger than the read length,
there would still be exponentially many solutions.
The culprit, as before, is repeats. 
Both $C$ and $E$ are repeats, appearing at least twice in the solution.
While we are now able to avoid over-collapsing them, 
we nevertheless are not able to determine the order in which the segments between them appear.

At this point, we realize a fundamental limitation of Problem~\ref{pr:bio}.
No matter what we do, if the genome has a repetitive structure, we will never be able to recover it.
Our problem seems unsolvable! 
Should we abandon it then, in frustration? 
Even quit computer science altogether?
Maybe, but before we do that, let us try one more thing.

We attempted to find a {\em genome reconstruction} from $S$ and $k$, 
which we defined as the string spelled by the shortest edge-covering walk of $DB^k(S)$. 
Though there were many such genome reconstructions, they were all based on the reordering of some more ``basic'' 
sequences.
We formalize this notion by saying that a string $c$ is {\em safe}, with respect to $S$ and $k$,
if, for any genome reconstruction $g$, $c$ is a substring of $g$.
We can now relax our requirements on the solution, 
with the hope that it may be more attainable to recover safe strings instead of a genome reconstruction
\begin{problem}[safe contigs]\label{pr:safecomplete}
	\hspace{1pt} \\
	\textbf{ Input:} A collection of strings $S=\{s_1, \ldots, s_m\}$, and an integer parameter $k$.\\
		\textbf{ Output:} The largest collection of strings that are safe for $S$ and $k$.
\end{problem}
Let us refer to a string in the output as a {\em contig}.
Actually, the ``shortest'' requirement in the notion of genome reconstruction has proven to be nothing but trouble and has not ultimately helped, so let us drop it.
That is, let us consider the above problem with respect to {\em genome reconstructions} which are strings spelled by edge-covering walks of $DB^k(S)$. 

We have implicitly defined an optimization problem,
where a feasible solution is a set of safe strings and the optimization criterion is the number of contigs.
Unlike~\Cref{pr:scs}, we are maximizing (as opposed to minimizing) the optimization criterion.
Note that without the optimization criterion, the formulation would not be useful, because
one feasible solution is to simply output every \kmer in $S$ as a separate safe string.
While these \kmers are certainly safe, there are usually many more safe strings.

\section{The unitig and omnitig algorithms}
To obtain an algorithm for~\Cref{pr:safecomplete}, let us try to draw on the intuition from our examples.
Consider~\Cref{fig:bubbles}.
We can see that any edge-covering walk of the graph should visit $C$ (or any other of the labeled walks, for that matter).
After all, once a walk enters the start of $C$ it has no other choice but to continue through to the end of $C$.
Let us try to capture the property of $C$ which makes it safe.
A {\em unitig} is a path such that all vertices except the first one have one incoming edge, and all vertices except the last one have one outgoing edge. 
A single vertex is always considered a unitig.
A {\em maximal unitig} is a unitig that cannot be extended in either direction.
It can be shown that the maximal unitigs of a graph are a partition of the vertex set, and this partition is unique\footnote{Recall that 
a {\em partition} of a set is a set of non-empty subsets such that every element is a member of exactly one subset.}.

In~\Cref{fig:bubbles}, the maximal unitigs correspond to the paths labeled by letters.
In~\Cref{fig:dbgscs}, the maximal unitigs are 
\begin{align*}
	(\dna{AT}, \dna{TT}, &\dna{TC}, \dna{CC}, \dna{CA}, \dna{AG}) \\
	& (\dna{AA}) \\
	& (\dna{GT}) \\
	(\dna{GC}, &\dna{CT}, \dna{TG}, \dna{GA}) 
\end{align*}

Based on this intuition, we define the {\em unitig algorithm}, which outputs the strings spelled by all the maximal unitigs in the graph.
It is easy to implement such an algorithm in polynomial time, but does it solve Problem~\ref{pr:safecomplete}?
Recall that a {\em source} is a vertex with no incoming edges, and a {\em sink} is a vertex with no outgoing edges. 
The good news is that if there is at least one genome reconstruction, source and sink, we can prove that maximal unitigs are always safe.
\begin{theorem}
	Let $S$ be a collection of strings, and let $k>1$ be an integer parameter.
	Suppose that $DB^k(S)$ contains at least one edge-covering walk, at least one source, and at least one sink.
	Then a unitig in $DB^k(S)$ spells a string that is safe for $S$ and $k$.
\end{theorem}
\begin{proof}
	Let $u=(v_1, \ldots, v_p)$ be a unitig, and let $g$ be an edge-covering walk of $DB^k(S)$.
	Note that $g$ must begin with a source and end with a sink.
	By the properties of the de Bruijn graph, 
	$u$ spells a string that is a substring of the string spelled by $g$ if and only if $u$ is a subwalk of $g$.
	We will therefore prove, by induction, that $u$ is a subwalk of $g$.
	In the base case, a unitig of a single vertex is trivially a subwalk of $g$.
	In general, the induction hypothesis gives us that $(v_1, \ldots, v_{p-1})$ is a subwalk of $g$.
	Since $v_{p-1}$ has an outgoing edge to $v_p$, it is not a sink.
	Hence, $g$ does not end at $v_{p-1}$. 
	The only possible vertex that can follow $(v_1, \ldots, v_{p-1})$ in $g$ is $v_p$, 
	because $v_{p-1}$ has only one out-neighbor.
	Thus, $(v_1, \ldots, v_{p-1}, v_p)$ is a subwalk of $g$.
\end{proof}

\begin{sloppypar}
	The bad news is that the unitig algorithm is not an optimal solution to~\Cref{pr:safecomplete}.
For example, in~\Cref{fig:dbgscs}, 
$(\dna{AA},\dna{AT},\dna{TT},\dna{TC},\dna{CC},\dna{CA},\dna{AG},\dna{GC},\dna{CT},\dna{TG},\dna{GA},$
$\dna{AT},\dna{TT},\dna{TC},\dna{CC},\dna{CA},\dna{AG})$
is a safe contig but not a unitig.
In~\Cref{fig:bubbles}, $CD_1$ is an example of a safe contig that is not a unitig.
Thus, the unitig algorithm is only a heuristic for Problem~\ref{pr:safecomplete}.
It finds a feasible solution but not an optimal one.
\end{sloppypar}

Nevertheless, 
we can still use the strategy of~\Cref{sec:exp} to check the usefulness of the unitig algorithm.
As we have stressed throughout, the ultimate test of an algorithm's usefulness 
is not whether it solves our formulation optimally, but whether it performs well in experimental evaluations.
Therefore, we proceed with the first stage of evaluation and run the unitig algorithm
on our simulated influenza data.
Eureka!
All the contigs are indeed substrings of the influenza genome, and most contigs are fairly long!

We could then proceed with the second stage of evaluation by simulating reads with gaps in coverage.
The unitig algorithm would still produce long and accurate contigs, though we may occasionally see a contig which 
is not a substring of the influenza genome.
Because such contigs would be rare, we would tolerate them and move on to introducing erroneous reads into our simulations.
We would find that the contigs are now short and inaccurate. 
Our choice of the first option from~\Cref{sec:simplify}, it turns out, was a mistake. 
But, as we hinted previously, we could still switch to the third option by formulating and solving the problem of 
correcting the errors in the reads. 
If we do that, we would find that the combination of an error correction algorithm and the unitig algorithm 
would perform well in simulations with errors and coverage gaps.
Moreover, this combination would also perform well on real data (i.e. the third state of evaluations).
In summary, the unitig algorithm, when combined with error correction, performs well in experimental evaluations 
and is robust to errors, coverage gaps, and other aspects of real data which we have not mentioned here.
It is a useful algorithm for solving our original~\Cref{pr:bio}.

Simultaneously, a careful investigation of~\Cref{pr:safecomplete} reveals that there is in fact an optimal polynomial time algorithm.
The algorithm is called the {\em omnitig algorithm}.
We will not cover how it works here, since it is fairly involved. 
It might seem that because the omnitig algorithm is optimal, it is strictly better than the unitig algorithm.
However, we must not give absolute credence to the superiority of a solution just because it is optimal for our problem formulation.
Because the omnitig algorithm has only been recently proposed, 
only the first stage of evaluation has been performed.
Whether or not it is ultimately more useful than the unitig algorithm will depend 
on how it fares when subjected to the second and third stages of evaluation.

\section{Conclusion}
By now, the reader has hopefully grown to appreciate the process of applying computer science techniques to a biological problem.
This process can 
happen within the context of a single paper, but, more typically, it spans many years.
Often, biologists do not have years to analyze their data, so 
``quick-and-dirty'' algorithms are used.
Such algorithms are often not solving a well-formulated problem and make it hard to draw on known computer science techniques.
Nevertheless, they can perform well in practice.
As a result, the process in this paper usually occurs only for those 
challenging problems where a quick-and-dirty approach does not work well enough.

The process described here requires an individual or a team that can understand both biology and computer science.
A breadth of computer science knowledge is needed to guide the designer to a formulation that is efficiently solvable,
while an understanding of the biology is needed to guide the designer to a formulation that performs well in experimental evaluations. 
Theoretical computer scientists sometimes ask if there are any difficult unsolved problems that are of interest to bioinformaticians. 
However, such problems are rare.
The difficult part is, more often than not, coming up with the right question, rather than the right answer.

\section{Notes}
This tutorial forms the basis of a series of lectures in a course on algorithms and data structures in bioinformatics at the Pennsylvania State University.
It is suitable for both undergraduate and graduate level.

The tutorial is not meant as a survey of genome assembly algorithms; in fact many important elements and
alternate approaches are not discussed here,
such as repeat resolution using paired-end reads, long-read technologies, or the string graph approach.
Real genome assembly algorithms go far beyond those in this tutorial, though unitigs are at the basis of many of them.
The process of experimental evaluation, which we only touched upon here, has its own set of challenges that need to be overcome 
(e.g. what are the right metrics for measuring accuracy and quality).
The idealized step-by-step developments in this presentation are also not meant to reflect the historical development of assembly algorithms,
and, hence, citations to related literature are not included.
There are several surveys describing assembly algorithms~\cite{el2013next,miller2010assembly,nagarajan2013sequence,schatz2010assembly,simpson2015theory,sohn2016present,wajid2012review}.

The organization and material of this tutorial is based on the constant exchange of thoughts with my colleagues,
either in conversations or through published work.
I am very grateful for their insights and open discussions,
though, as these have become integrated into my own understanding, 
it is impossible to disentangle them and assign proper credit to specific individuals.
I am also grateful for early feedback on the manuscript from Alexandru Tomescu, Jens Stoye, and Adam D. Smith, 
as well as feedback from anonymous reviewers.
This work has been supported in part by
NSF awards DBI-1356529, CCF-1439057, IIS-1453527, and IIS-1421908.

\bibliographystyle{plain}
\bibliography{assemblyLectureNotes}

\end{document}